\documentclass{article}

\usepackage{fullpage}

\usepackage[
normalsections, 
normalmargins, 
normallists, 
normalfloats, 
normalindent, 
normaltitle, 
normalleading, 
normallooseness, 
]{savetrees}

\usepackage{tikz}
\usetikzlibrary{snakes}

\usepackage{amsmath,amstext,amssymb,amsthm}
\usepackage{xspace}   

\newcommand{\defproblemu}[3]{
  \vspace{1mm}
\noindent\fbox{
  \begin{minipage}{0.98\textwidth}
  #1 \\
  {\bf{Input:}} #2  \\
  {\bf{Question:}} #3
  \end{minipage}
  }
  \vspace{1mm}
}
\newcommand{\defparproblemu}[4]{
  \vspace{1mm}
\noindent\fbox{
  \begin{minipage}{0.98\textwidth}
  \begin{tabular*}{0.98\textwidth}{@{\extracolsep{\fill}}lr} #1 & {\bf{Parameter:}} #3 \\ \end{tabular*}
  {\bf{Input:}} #2  \\
  {\bf{Question:}} #4
  \end{minipage}
  }
  \vspace{1mm}
}

\newcommand{\mcname}{{\sc{Multiway Cut}}\xspace}
\newcommand{\nmcname}{{\sc{Node}} \mcname\xspace}
\newcommand{\nmcacutname}{{\sc{NMWC-a-Cut}}\xspace}
\newcommand{\nmcalpname}{{\sc{NMWC-a-LP}}\xspace}
\newcommand{\vcname}{{\sc{Vertex Cover}}\xspace}
\newcommand{\vcammname}{\vcname{} {\sc{above Maximum Matching}}\xspace}
\newcommand{\asatname}{{\sc{Almost 2-SAT}}\xspace}
\newcommand{\nmulticutname}{{\sc{Node Multicut}}\xspace}
\newcommand{\nmulticutalpname}{{\sc{NMC-a-LP}}\xspace}

\newcommand{\pp}{{L}}
\newcommand{\lpcon}{{C}}

\newtheorem{theorem}{Theorem}[section]
\newtheorem{lemma}[theorem]{Lemma}
\newtheorem{definition}[theorem]{Definition}
\newtheorem{corollary}[theorem]{Corollary}

\newtheorem{reduction}{Reduction}
\newtheorem*{branch}{Branching rule}

\begin{document}
  \author{Marek Cygan\thanks{Institute of Informatics, University of Warsaw, Poland, \texttt{cygan@mimuw.edu.pl}} \and
	Marcin Pilipczuk\thanks{Institute of Informatics, University of Warsaw, Poland, \texttt{malcin@mimuw.edu.pl}} \and
	Micha\l{} Pilipczuk\thanks{Faculty of Mathematics, Informatics and Mechanics, University of Warsaw, Poland, \texttt{michal.pilipczuk@students.mimuw.edu.pl}} \and
	Jakub Onufry Wojtaszczyk\thanks{Google Inc., Cracow, Poland, \texttt{onufry@google.com}}}
\date{}
\title{On Multiway Cut parameterized above lower bounds}

\maketitle

\begin{abstract}
 In this paper we consider two {\em{above lower bound}} parameterizations
 of the \nmcname problem --- above the maximum separating cut and 
 above a natural LP-relaxation ---
 and prove them to be fixed-parameter tractable.
 Our results imply $O^*(4^k)$ algorithms for \vcammname and \asatname
 as well as an $O^*(2^k)$ algorithm for \nmcname
 with a standard parameterization by the solution size,
 improving previous bounds for these problems.
\end{abstract}

\section{Introduction}

The study of cuts and flows is one of the most active fields in combinatorial optimization.
However, while the simplest case, where we seek a cut separating two given vertices of a graph,
is algorithmically tractable, the problem becomes hard as soon as one starts to deal
with multiple terminals.
For instance, given three vertices in a graph it is NP-hard to decide what is the smallest size
of a cut that separates every pair of them (see \cite{nmc-np-hardness}).
The generalization of this problem --- the well-studied \nmcname problem ---
asks for the size of the smallest set
separating a given set of terminals. The formal definition is as follows:

\defproblemu{\nmcname}{A graph $G=(V,E)$, a set $T \subseteq V$ of {\em{terminals}}
  and an integer $k$.}{Does there exist a set $X \subseteq V \setminus T$ of size at most
    $k$ such that any path between two different terminals intersects $X$?}

For various approaches to this problem we refer the reader for instance to
\cite{garg:multiway-approx,chen-nmc,nmc-np-hardness,marx-separators}.

Before describing our results, let us discuss the methodology we will be working with.
We will be studying \nmcname (and several other problems) from the parameterized
complexity point of view. Note that since the solution to our problem is a set of $k$ vertices
and it is easy to verify whether a solution is correct, we can solve the problem by
enumerating and verifying all the $O(|V|^k)$ sets of size $k$. Therefore, for every fixed
value of $k$, our problem can be solved in polynomial time.
This approach, however, is not feasible even for, say, $k = 10$. The idea of
parameterized complexity is to try to split the (usually exponential) dependency on $k$
from the (hopefully uniformly polynomial) dependency on $|V|$ --- 
so we look for an algorithm where the degree of the polynomial does not
depend on $k$, e.g., an $O(C^k |V|^{O(1)})$ algorithm.

Formally, a parameterized problem $Q$ is a subset of $\Sigma^* \times \mathbb{N}$
for some finite alphabet $\Sigma$, where the integer is the {\em parameter}.
We say that the problem is {\em{fixed parameter tractable}} ({\em{FPT}})
if there exists an algorithm solving any instance $(x,k)$ in
time $f(k) {\rm poly}(|x|)$ for some (usually exponential) computable function $f$.
It is known that a problem is FPT iff
it is kernelizable: a kernelization algorithm for a problem $Q$ takes an instance $(x,k)$ and in
time polynomial in $|x|+k$ produces an equivalent instance $(x', k')$ (i.e., $(x,k)\in Q$ iff $(x',k')\in Q$) such that
$|x'| + k' \leq g(k)$ for some computable function $g$.
The function $g$ is the {\em{size of the kernel}}, and if it is polynomial, we say that $Q$ admits a polynomial kernel.
The reader is invited to refer
to now classical books by Downey and Fellows \cite{downey-fellows:book},
Flum and Grohe \cite{grohe:book} and Niedermeier \cite{niedermeier:book}.

The typical parameterization takes the solution size as the parameter.
For instance, Chen et al. \cite{chen-nmc} have shown an
algorithm solving \nmcname in time $O(4^k n^{O(1)})$, improving upon the previous result of Daniel Marx \cite{marx-separators}.
However, in many cases it turns out we have a natural lower bound on the solution size ---
for instance, in the case of the \vcname problem the cardinality of the maximal matching is
such a lower bound. It can happen that this lower bound is large --- rendering algorithms
parameterized by the solution size impractical. 
For some problems, better answers have been obtained by introducing the so called
{\em parameterization above guaranteed value}, i.e. taking as the parameter the
difference between the expected solution size and the lower bound. The idea was 
first proposed in~\cite{param-above}. An overview of this
currently active research area can be found in the introduction to~\cite{gutin-belkot}.

We will consider two natural lower bounds for \nmcname --- the separating cut and the
LP-relaxation solution.
Let $I=(G,T,k)$ be a \nmcname{} instance and let $s = |T|$.
By a {\em{minimum solution}} to $I$ we mean a set $X \subseteq V \setminus T$ of minimum
cardinality that disconnects the terminals, even if $|X| > k$.

For a terminal $t \in T$ a set $S \subseteq V \setminus T$ is a {\em{separating cut}}
of $t$ if $t$ is disconnected from $T \setminus \{t\}$ in $G[V \setminus S]$
(the subgraph induced by $V \setminus S$).
Let $m(I,t)$ be the size of a minimum isolating cut of $t$.
Notice that for any $t$ the value $m(I, t)$ can be found in polynomial time using
standard max-flow techniques. Moreover, this value is a lower bound for the size of the
minimum solution to $I$ --- any solution $X$ has, in particular, to separate $t$ from
all the other terminals.

Now we consider a different approach to the problem, stemming from linear programming.
Let $\mathcal{P}(I)$ denote the
set of all simple paths connecting two different terminals in $G$.
Garg et al. \cite{garg:multiway-approx} gave a $2$-approximation algorithm using the following
natural LP-relaxation:
\begin{align}\label{eq:lp}
\textrm{minimize}\  & \sum_{v \in V \setminus T} d_v \\
\textrm{subject to}\  & \sum_{v \in P \cap (V \setminus T)} d_v \geq 1 & \forall P \in \mathcal{P}(I)\nonumber\\
  & d_v \geq 0 & \forall v \in V \setminus T\nonumber
\end{align}
In other words, the LP-relaxation asks to assign for each vertex $v \in V \setminus T$
a non-negative weight
$d_v$, such that the distance between pair of terminals, with respect to the weights $d_v$,
is at least one. This is indeed a relaxation of the original problem --- if we
restrict the values $d_v$ to be integers, we obtain the original \nmcname.

The above LP-relaxation has exponential number of constraints, as $\mathcal{P}(I)$ can be exponentially
big in the input size.
However, the optimal solution for this LP-relaxation can be found in
polynomial time either using separation oracle and ellipsoid method
or by solving an equivalent linear program of polynomial size (see \cite{garg:multiway-approx} for details).
By $LP(I)$ we denote the cost of the optimal solution of the LP-relaxation (\ref{eq:lp}).
As the LP-relaxation is less restrictive than the original \nmcname problem, $LP(I)$ is
indeed a lower bound on the size of the minimum solution.

We can now define two {\em{above lower bound}}
parameters: $\pp(I) = k - LP(I)$ and $\lpcon(I) = k - \max_{t \in T} m(I,t)$,
and denote by \nmcalpname{} ({\sc{Node Multiway Cut above LP-relaxation}}) and \nmcacutname{} ({\sc{Node Multiway Cut above Maximum Separating Cut}})
the \nmcname{} problem parameterized by $\pp(I)$ and $\lpcon(I)$, respectively.

We say that a parameterized problem $Q$ is in $XP$,
if there exists an algorithm solving any instance $(x,k)$ in time $|x|^{f(k)}$
for some computable function $f$, i.e., polynomial for any constant value of $k$.
The \nmcacutname{} problem was defined and shown to be in 
XP by Razgon in~\cite{razgon:arxiv2010}.

\paragraph{Our results}

In Section \ref{sec:alg},
using the ideas of Xiao~\cite{xiao:multiway2010}
and building upon analysis of the LP relaxation by Guillemot \cite{guillemot:iwpec2008},
we prove a \nmcname{} instance $I$
can be solved in $O^*(4^{\pp(I)})$ time\footnote{$O^*()$ is the $O()$ notation with suppressed
  factors polynomial in the size of the input.},
which easily yields an $O^*(2^{\lpcon(I)})$-time
algorithm. Both algorithms run in polynomial space.
Consequently we prove both \nmcalpname{} and \nmcacutname{} problems to be FPT,
solving an open problem of Razgon~\cite{razgon:arxiv2010}.
Observe that if $\lpcon(I) > k$ the answer is trivially negative,
hence as a by-product we obtain
an $O^*(2^k)$ time algorithm for the \nmcname{} problem, improving the previously best known
$O^*(4^k)$ time algorithm by Chen et al.~\cite{chen-nmc}.

By considering a line graph of the input graph, it is easy to see that an
edge-deletion variant of {\sc{Multiway Cut}} is easier than the node-deletion one,
and our results hold also for the edge-deletion variant.
We note that the edge-deletion variant, parameterized above maximum separating cut,
was implicitly proven to be FPT by Xiao \cite{xiao:multiway2010}.

Furthermore we observe that \vcammname{} is a special case of \nmcalpname{},
while it is known that \vcammname{} is equivalent to \asatname{} from
the point of view of parameterized complexity \cite{gutin:vc-a-mm,razgon:icalp2008}.
The question of an FPT algorithm for those two problems was a long-standing open
problem until Razgon and O'Sullivan gave an $O^*(15^k)$-time algorithm in 2008.
Our results improve this bound to $O^*(4^k)$ for both \vcammname{} and \asatname{}.
The details are gathered in Section \ref{sec:2sat}.

One of the major open problems in kernelization is the question of a polynomial
kernel for \nmcname{}, parameterized by the solution size.
Our results show that the number of terminals can be reduced to $2k$ in polynomial-time, improving a quadratic
bound due to Razgon \cite{razgon:arxiv2011}. Moreover, our algorithm includes a number of
polynomial-time reduction rules, that may be of some interest from the point of view of kernelization.

Finally, we consider the \nmulticutname{} problem, a generalization of \nmcname{},
which was recently proven to be FPT when parameterized by the solution size \cite{marx:multicut,thomasse:multicut}.
In Section \ref{sec:multicut} we show that \nmulticutname{}, when parameterized above a natural LP-relaxation,
   is significantly more difficult and even not in $XP$.

\paragraph{Notation}

Let us introduce some notation. All considered graphs are undirected and simple.
Let $G=(V,E)$ be a graph. For $v \in V$ by $N(v)$ we denote the set of neighbours of
$v$, $N(v) = \{u \in V: uv \in E\}$, and by $N[v]$ the closed neighbourhood of $v$,
$N[v] = N(v) \cup \{v\}$. We extend this notation to subsets of vertices $S \subseteq V$,
$N[S] = \bigcup_{v \in S} N[v]$, $N(S) = N[S] \setminus S$.
By {\em{removing}} a vertex $v$ we mean transforming $G$ to
$(V \setminus v, E \setminus \{uv, vu : u \in V\})$.
The resulting graph is denoted by $G \setminus v$.
By {\em{contracting an edge}} $uv$ we mean the following operation: we remove vertices
$u$ and $v$, introduce a new vertex $x_{uv}$ and connect it to all vertices previously
connected to $u$ or $v$. The resulting graph is denoted by $G \slash uv$.
If $u \in T$ and $v \notin T$, we somewhat abuse the notation
and identify the new vertex $x_{uv}$ with $u$, so that the terminal set
remains unchanged. In this paper we do not contract any edge that connects two terminals.

\section{Algorithms for \mcname{}}\label{sec:alg}

Let $I=(G,T,k)$, where $G = (V,E)$, be a \nmcname{} instance.
First, let us recall the two known facts about the LP-relaxation (\ref{eq:lp}).
\begin{definition}[\cite{garg:multiway-approx,guillemot:iwpec2008}]
Let $(d_v)_{v \in V \setminus T}$ be a feasible solution to the LP-relaxation (\ref{eq:lp}) of $I$.
For a terminal $t$, the {\em{zero area}} of $t$, denoted by $U_t$, is the set of vertices within distance zero from $t$ with respect to weights $d_v$.
\end{definition}
\begin{lemma}[\cite{garg:multiway-approx}]\label{lem:garg}
Given an optimal solution $(d_v^*)_{v \in V \setminus T}$
to the LP-relaxation (\ref{eq:lp}), let us construct an assignment $(d_v)_{v \in V \setminus T}$
as follows. First, for each terminal $t$ compute its zero area $U_t$ with respect to weights $(d_v^*)_{v \in V \setminus T}$.
Second, for $v \in V \setminus T$
we take $d_v = 1$ if $v \in N(U_t)$
for at least two terminals $t$, $d_v = 1/2$ if $v \in N(U_t)$
for exactly one terminal $t$, and $d_v = 0$ otherwise.
Then $(d_v)_{v \in V \setminus T}$ is also an optimal solution to the LP-relaxation (\ref{eq:lp}).
\end{lemma}
\begin{lemma}[\cite{guillemot:iwpec2008}, Lemma 3]\label{lem:guillemot}
Let $(d_v^*)_{v \in V \setminus T}$ be any optimal solution to the LP-relaxation (\ref{eq:lp}) of $I$.
Then there is a minimum solution to $I$ that is disjoint with $\bigcup_{t \in T} U_t$.
\end{lemma}

Our algorithm consists of two parts. The first part is a set of several polynomial-time
reduction rules. At any moment, we apply the lowest-numbered applicable rule.
We shall prove that the original instance $I$ is a YES--instance if and only
if the reduced instance is a YES--instance, we will say this means the
reduction is {\em sound}.
We prove that no reduction rule increases the parameter $\pp(I)$ or the graph size.
If no reduction rule can be applied, we proceed to the branching rule.
The branching rule outputs two subcases, each with the parameter $\pp(I)$ decreased
by at least $1/2$ and a smaller graph. If the answer to any of the two subcases is YES,
we return YES from the original instance, otherwise we return NO.
As the parameter $\pp(I)$ decreases by at least $1\slash 2$ with each branching,
and we can trivially return NO if $\pp(I)$ is negative, 
we obtain the claimed $O^*(4^{\pp(I)})$ running time.

\begin{reduction}\label{red:no}
  If two terminals are connected by an edge or $\pp(I) < 0$, return NO.
\end{reduction}
The first part of the above rule is obviously sound, as we only remove vertices, not edges.
The second part is sound as the optimal cost of the LP-relaxation (\ref{eq:lp})
is a lower bound for the size of the minimum solution to the instance $I$.

\begin{reduction}\label{red:common}
  If there exists a vertex $w \in V \setminus T$ that is adjacent to two
  terminals $t_1, t_2 \in T$, remove $w$ from $G$ and decrease $k$ by one.
\end{reduction}
The above rule is sound, as such a vertex $w$ has to be included in any solution
to $I$. Let us now analyse how the parameter $\pp(I)$ is influenced by this rule.
Let $I' = (G \setminus w,T,k-1)$ be the output instance.
Notice that any feasible solution $(d_v)_{v \in V \setminus (T \cup \{w\})}$ to $I'$ can be
extended to a feasible solution of $I$ by putting $d_w = 1$. Thus $LP(I) \leq LP(I') + 1$, and
we infer $\pp(I) \geq \pp(I')$.

\begin{reduction}\label{red:eat}
  Let $w \in V \setminus T$ be a neighbour of a terminal $t \in T$.
  Let $(d_v^\circ)_{v \in V \setminus T}$ be a solution to the LP-relaxation
  (\ref{eq:lp}) with an additional constraint $d_w=0$. If the cost
  of the solution $(d_v^\circ)_{v \in V \setminus T}$ is equal to $LP(I)$,
  contract the edge $tw$.
\end{reduction}
As $(d_v^\circ)_{v \in V \setminus T}$ is a feasible solution to
the LP-relaxation (\ref{eq:lp}), its cost is at least $LP(I)$.
If the rule is applicable, $(d_v^\circ)_{v \in V \setminus T}$
is an optimal solution to the LP-relaxation (\ref{eq:lp})
and $w \in U_t$. The soundness of Reduction \ref{red:eat} follows from Lemma \ref{lem:guillemot}.
Moreover, note that if $I'$ is the output instance of Reduction \ref{red:eat},
we have $LP(I) = LP(I')$, as $(d_v^\circ)_{v \in V \setminus (T \cup \{w\})}$
is a feasible solution to the LP-relaxation (\ref{eq:lp}) for the instance $I'$.
We infer that $\pp(I) = \pp(I')$.

The following lemma summarizes properties of an instance, assuming none of the above
reduction rules is applicable.
\begin{lemma}\label{lem:norule}
If Reductions \ref{red:no}, \ref{red:common} and \ref{red:eat} are not applicable, then:
\begin{enumerate}
\item An assignment $(d_v)_{v \in V \setminus T}$ that assigns $d_v = 1/2$ if $v \in N(T)$
and $d_v = 0$ otherwise is an optimal solution to the LP-relaxation (\ref{eq:lp}).
\item For each terminal $t \in T$, the set $N(t)$ is the unique minimum separating cut of $t$.
\end{enumerate}
\end{lemma}
\begin{proof}
Let $(d_v^*)_{v \in V \setminus T}$ be any optimal solution to the LP-relaxation
(\ref{eq:lp}). As Reduction \ref{red:eat} is not applicable, $d_w^* > 0$ for any $w \in N(T)$.
As Reduction \ref{red:common} is not applicable, if we invoke Lemma \ref{lem:garg}
on the assignment $(d_v^*)_{v \in V \setminus T}$, we obtain the assignment
$(d_v)_{v \in V \setminus T}$. Thus the first part of the lemma is proven.

For the second part, obviously $N(t)$ is a separating cut of $t$.
Let $C \subseteq V \setminus T$ be any other separating cut of $t$ and assume $|C| \leq |N(t)|$.
Let $d_v' = d_v + 1/2$ if $v \in C \setminus N(t)$, $d_v' = d_v - 1/2$ if $v \in N(t) \setminus C$
and $d_v' = d_v$ otherwise. It is easy to see that $d_v'$ is a feasible solution to
the LP-relaxation (\ref{eq:lp}). As $|C| \leq |N(t)|$,
    $\sum_{v \in V \setminus T} d_v' \leq \sum_{v \in V \setminus T} d_v$
    and we infer that $(d_v')_{v \in V \setminus T}$ is an optimal solution
to the LP-relaxation (\ref{eq:lp}). However, $d_v' = 0$ for $v \in N(t) \setminus C$,
and Reduction \ref{red:eat} would be applicable.
\end{proof}

\begin{branch}
  Let $w \in V \setminus T$ be a neighbour of a terminal $t \in T$.
  Branch into two subcases, either $w$ is included in a solution to the \nmcname{}
  instance $I$ or not. In the first branch, we remove $w$ from the graph and decrease $k$
  by one. In the second one, we contract the edge $tw$.
\end{branch}

The soundness of the branching rule is straightforward. We now prove that in both subcases
the parameter $\pp(I)$ drops by at least $1/2$. Let $I_1 = (G \setminus w, T, k-1)$
and $I_2 = (G \slash tw, T, k)$ be the output instances in the first and second cases,
respectively.

In the first subcase, it is sufficient to prove that $LP(I_1) \geq LP(I) - 1/2$, i.e., that
the cost of the optimal solution to the LP-relaxation (\ref{eq:lp}) drops by at most half.
Assume the contrary, that $LP(I_1) < LP(I) - 1/2$.
Let $(d_v)_{v \in V \setminus T \setminus \{w\}}$ be a half-integral optimal
solution to the LP-relaxation (\ref{eq:lp}) for $I_1$, as asserted
by Lemma \ref{lem:garg}. Note that if we put $d_w = 1$, then $(d_v)_{v \in V \setminus T}$
is a feasible solution to the LP-relaxation (\ref{eq:lp}) for $I$, and $LP(I_1) \geq LP(I)-1$.
By half-integrality, $LP(I_1) = LP(I) - 1$ and $(d_v)_{v \in V \setminus T}$ is an optimal
half-integral solution to the LP-relaxation (\ref{eq:lp}) for $I$.
As Reduction \ref{red:eat} is not applicable, $d_v > 0$ for all $v \in N(T)$.
As $(d_v)_{v \in V \setminus T}$ is half-integral, $d_v \geq 1/2$ for all $v \in N(T)$.
However, the assignment given by Lemma \ref{lem:norule} has strictly smaller cost
than $(d_v)_{v \in V \setminus T}$ (as $d_w = 1$),
a contradiction to the fact that $(d_v)_{v \in V \setminus T}$ is an optimal
solution to the LP-relaxation (\ref{eq:lp}) for $I$.
Thus $LP(I_1) \geq LP(I) -1/2$ and $\pp(I_1) \leq \pp(I) - 1/2$.

In the second subcase note that, as Reduction \ref{red:eat} is not applicable,
$LP(I_2) > LP(I)$. As the LP-relaxation (\ref{eq:lp}) has half-integral solutions,
we have $LP(I_2) \geq LP(I) + 1/2$. This implies that $\pp(I_2) \leq \pp(I)-1/2$.

Since Reduction \ref{red:no} stops when the parameter $\pp(I)$ becomes negative,
we obtain the following theorem.
\begin{theorem}\label{thm:nmc-lp}
  There exists an algorithm that solves a \nmcname{} instance $I$
  in $O^*(4^{\pp(I)})$ time.
\end{theorem}

To solve \nmcname{} parameterized by $\lpcon(I)$, we introduce one more reduction rule.
Recall $s$ denotes the number of terminals.
\begin{reduction}\label{red:yes}
  If $\lpcon(I) \geq \frac{s-2}{s-1} \cdot k$ or $\lpcon(I) \leq 2\pp(I)$, return YES.
\end{reduction}
Now we show that Reduction~\ref{red:yes} is sound.
Let $t_0 \in T$ be the terminal with the largest separating cut, i.e.,
$m(I,t_0) = \max_{t \in T} m(I,t)$.
Let $X = N(T \setminus \{t_0\})$.
Obviously no two terminals are in the same connected component of $G[V \setminus X]$.
We claim that $|X| \leq k$.

If $\lpcon(I) \geq \frac{s-2}{s-1} \cdot k$, $|N(t)| = m(I,t)$ by Lemma \ref{lem:norule}, and:
$$|X| = \sum_{t \in T \setminus \{t_0\}} m(I,t) \leq (s-1)m(I,t_0) = (s-1)(k-\lpcon(I)) \leq k.$$

In the second case, the condition $\lpcon(I) \leq 2\pp(I)$ is equivalent to $2LP(I) - m(I,t_0) \leq k$.
From the structure of the optimum half-integral solution
given by Lemma \ref{lem:norule}, we infer that $2LP(I) \geq |N(T)|$.
By Lemma \ref{lem:norule}, $|N(t_0)| = m(I,t_0)$.
Since Reduction \ref{red:common} is not applicable,
$N(t_0) \cap N(t) = \emptyset$ for $t \in T \setminus \{t_0\}$.
We infer that $|X| = 2LP(I) - m(I,t_0) \leq k$,
and Reduction \ref{red:yes} is sound.
\begin{corollary}\label{cor:nmc-cut}
  There exists an algorithm that solves a \nmcname{} instance $I$
  in $O^*(2^{\min(\lpcon(I), \frac{s-2}{s-1} \cdot k)})$ time.
  In the case of three terminals, this yields a $O^*(2^{k/2})$-time algorithm.
\end{corollary}

Finally, we would like to note that all our reduction rules
are polynomial-time and could be used in a hypothetical algorithm to
find a polynomial kernel for \nmcname{}.
Let us supply them with one additional clean-up rule.
\begin{reduction}\label{red:cleanup}
If there exists a connected component of $G$ with at most one terminal, remove it.
\end{reduction}
The following lemma shows that our reductions improve
the quadratic bound on the number of terminals due to Razgon \cite{razgon:arxiv2011}.
\begin{lemma}\label{lem:2kterm}
If Reductions \ref{red:no}, \ref{red:common}, \ref{red:eat} and \ref{red:cleanup}
are not applicable, then $|T| \leq 2k$.
\end{lemma}
\begin{proof}
As noted before, the optimal half-integral solution given by Lemma \ref{lem:norule}
implies that $|N(T)| = 2LP(I)$.
However, if Reduction \ref{red:cleanup} is not applicable,
$N(t) \neq \emptyset$ for any $t \in T$, and $|T| \leq |N(T)|$ by Reduction \ref{red:common}.
We infer that $2LP(I) \geq |T|$.
If $|T| > 2k$, then $\pp(I) < 0$ and Reduction \ref{red:no} would return NO.
\end{proof}

\section{From \nmcname{} to \asatname{}}\label{sec:2sat}

We start with problem definitions.
For a graph $G$ by $\mu(G)$ we denote the size of a maximum matching in $G$.

\defparproblemu{\vcammname{}}{A graph $G=(V,E)$ and an integer $k$.}{$k$}{Does there exist a vertex cover in $G$ of size at most $\mu(G) + k$?}

\defparproblemu{\asatname{}}{A $2$-SAT formula $\Phi$ and an integer $k$.}{$k$}{
  Does there exist a set $X$ of at most $k$ clauses of $\Phi$, whose deletion
    makes $\Phi$ satisfiable?}

Now we prove that \vcammname{} is a special case of \nmcalpname{}.

\begin{theorem}\label{thm:vcamm}
There exists an algorithm that solves \vcammname{} in $O^*(4^k)$ time.
\end{theorem}
\begin{proof}
Let $I=(G=(V,E),k)$ be a \vcammname{} instance. We construct a \nmcname{} instance
$I'=(G',T,k')$ as follows. For each $v \in V$ we create a terminal $t_v$
and connect it to $v$, thus $T = \{t_v : v \in V\}$ and each terminal in $G'$
is of degree one. Moreover we take $k' = \mu(G) + k$.

We claim that $X \subseteq V$ is a vertex cover in $G$ if and only if each
connected component of $G'[(V \setminus X) \cup T]$ contains at most one terminal.
If $X \subseteq V$ is a vertex cover in $G$, $G[V \setminus X]$ is an independent
set, thus every edge in $G'[(V \setminus X) \cup T]$ is of type $t_vv$.
In the other direction, note that if $uv \in E$ and $u,v \notin X$, then
$t_u$ and $t_v$ are connected in $G'[(V \setminus X) \cup T]$.

We now show that $LP(I') \geq \mu(G)$.
Let $M$ be a maximum matching in $G$ and let $(d_v)_{v \in V}$ be an optimal
solution to the LP-relaxation (\ref{eq:lp}) for $I'$.
For each $uv \in M$, the path consisting of vertices $t_u$, $u$, $v$ and $t_v$
is in $\mathcal{P}(I')$, thus $d_u + d_v \geq 1$. As $M$ is a matching,
we infer that $\sum_{v \in v} d_v \geq |M| = \mu(G)$.

Since $LP(I') \geq \mu(G)$ and $k' = k + \mu(G)$, we have $\pp(I') \leq k$.
We apply algorithm from Theorem \ref{thm:nmc-lp} to the instance $I'$
and the time bound follows.
\end{proof}

We now reproduce the reduction from \asatname{} to \vcammname{} to prove the following
theorem.

\begin{theorem}\label{thm:asat}
There exists an algorithm that solves \asatname{} in $O^*(4^k)$ time.
\end{theorem}
\begin{proof}
Let $I=(\Phi, k)$ be an \asatname{} instance.
First, we replace each clause $C \in \Phi$ that consists of a single literal $l$
with a clause $(l \vee l)$. From now we assume that each clause of $\Phi$ consists
of two, possibly equal, literals.

Let $x$ be a variable of $\Phi$. By $n(x)$ we denote the number of occurrences of the variable
$x$ in the formula $\Phi$ (if $l = x$ or $l = \neg x$, a clause $(l \vee l)$
counts as two occurrences). Let us arbitrarily number those occurrences
and for any $1 \leq i \leq n(x)$, by $C(x,i)$ we denote the clause where $x$ occurs
the $i$-th time.

We now construct a \vcammname{} instance $I'=(G,k)$.
For each variable $x$ and for each $1 \leq i \leq n(x)$ we create
two vertices $v(x,i)$ and $v(\neg x, i)$. For $l \in \{x,\neg x\}$ we denote
$V(l) = \{v(l,i): 1 \leq i \leq n(x)\}$. For each variable $x$
and for each $1 \leq i,j \leq n(x)$ we connect $v(x,i)$ and $v(\neg x, j)$ by an edge,
    i.e., we make a full bipartite subgraph with sides $V(x)$ and $V(\neg x)$.

Furthermore, if $C(x,i) = C(y,j)$ for some variables $x, y$ and indices
$1 \leq i \leq n(x)$, $1 \leq j \leq n(y)$ (possibly $x=y$, but $(x,i) \neq (y,j)$),
we introduce an edge $v(l_x,i)v(l_y,j)$, where $C(x,i) = C(y,j) = (l_x \vee l_y)$,
$l_x$ is the $i$-th occurrence of $x$ and $l_y$ is the $j$-th occurrence of $y$.
Such an edge is called a {\em{clause edge}}. Note that we introduce exactly one
clause edge for each clause of $\Phi$ and no two clause edges share an endpoint in $G$.

We claim that $I$ is an \asatname{} YES-instance if and only if $I'$ is a \vcammname{} YES-instance.
First note that $G$ has a perfect matching consisting of all edges of the type $v(x,i)v(\neg x, i)$.

Assume $I$ is a YES-instance.
Let $X \subseteq \Phi$ be a set of clauses, such that there exists a truth assignment $\phi$ of all variables of $\Phi$
that satisfies all clauses of $\Phi \setminus X$.
We now construct a vertex cover $Y$ of $G$. For each variable $x$ and for
each index $1 \leq i \leq n(x)$, we take into $Y$ the vertex $v(x,i)$ if $x$ is true
in the assignment $\phi$, and $v(\neg x,i)$ otherwise. Moreover, for each clause
$C \in X$ we take into $Y$ any endpoint of the clause edge for $C$.

Clearly $|Y| \leq \mu(G) + |X|$. Each non-clause edge $v(x,i)v(\neg x,j)$ is covered
by $Y$, as $v(x,i) \in Y$ if $x$ is true in $\phi$, and $v(\neg x,j) \in Y$ otherwise.
Let $e_C = v(l_x,i)v(l_y,j)$ be a clause edge for clause $C=(l_x \vee l_y)$.
If $C \in X$, then one of the endpoints of $e_C$ is chosen into $Y$.
Otherwise, $l_x$ or $l_y$ is true in $\phi$ and the corresponding vertex is chosen
into $Y$.

In the other direction, let us assume that $I'$ is a YES-instance
and let $Y$ be a vertex cover of $G$. We construct a truth assignment $\phi$
as follows. Let $x$ be a variable of $\Phi$.
Recall that $G$ has a complete bipartite subgraph with sides $V(x)$ and $V(\neg x)$.
Thus $V(l) \subseteq Y$ for some $l \in \{x,\neg x\}$, and we take $l$ to be true
in $\phi$ (if $V(x) \cup V(\neg x) \subseteq Y$, we choose whether $x$ is true or false
arbitrarily). Let $X$ be the set of clauses of $\Phi$ that are not satisfied by $\phi$.
We claim that $|X| \leq |Y| - \mu(G)$.

Let $Y_1$ be the union of all sets $V(l)$ for which $l$ is true under $\phi$.
Obviously $Y_1 \subseteq Y$ and $|Y_1| = \mu(G)$. Let $Y_2 = Y \setminus Y_1$.
Take any $C \in X$. As $C$ is not satisfied by $\phi$, the clause edge $e_C$
corresponding to $C$ does not have an endpoint in $Y_1$. Since $Y$ is a vertex cover
in $G$, $e_C$ has an endpoint in $Y_2$. Finally, recall that no two clause edges
share an endpoint. This implies that $|Y_2| \geq |X|$ and $|X| \leq |Y| - \mu(G)$.

We infer that the instances $I$ and $I'$ are equivalent. As the above construction
can be done in polynomial time, the running time follows from Theorem \ref{thm:vcamm}.
\end{proof}

\section{Hardness of \nmulticutname{} parameterized above LP-relaxation}\label{sec:multicut}

Recall the definition of \nmulticutname{}, which is a natural generalization of \nmcname{}.

\defproblemu{\nmulticutname{}}{A graph $G=(V,E)$, a set $\mathcal{T}$ of pairs of terminals, and an integer $k$.}{Does there exist a set $X$ of at most $k$ non-terminal vertices,
whose removal disconnects all pairs of terminals in $\mathcal{T}$?}

The LP-relaxation (\ref{eq:lp}) for \nmcname{} naturally generalizes to \nmulticutname{} as follows.
Let $T$ be the set of all terminals in the given \nmulticutname{} instance $I = (G,\mathcal{T},k)$.
In the LP-relaxation we ask for an assignment of non-negative weights $(d_v)_{v \in V \setminus T}$,
such that for each pair $(s,t) \in \mathcal{T}$ the distance between $s$ and $t$ with respect to the weights $(d_v)_{v \in V \setminus T}$
is at least one. Clearly, if $X$ is a solution to $I$, an assignment that takes $d_v = 1$ if $v \in X$ and $d_v = 0$ otherwise,
is a feasible solution to the LP-relaxation. Let $LP(I)$ be the cost of an optimal solution to this LP-relaxation.
We denote by \nmulticutalpname{} the \nmulticutname{} problem parameterized by $\pp(I)=k-LP(I)$, i.e., parameterized above LP lower bound.

In this section we prove that \nmulticutalpname{} does not even belong to $XP$, by the following lemma.
\begin{lemma}
\nmulticutalpname{}, restricted to instances where $\pp(I) = 0$, is NP-hard.
\end{lemma}
\begin{proof}
We reduce from {\sc{Multicoloured Independent Set}} which is NP-complete (see~\cite{mis-np-hard}). In this problem we are given a graph $G=(V,E)$ together with a partition of
the vertex set into sets $V_1, V_2, \ldots, V_r$, such that $G[V_i]$ is a clique for $1 \leq i \leq r$, and we are to decide
whether $G$ contains an independent set of size $r$. Note that such an independent set needs to take exactly one vertex from each set $V_i$.
W.l.o.g. we may assume that $|V_i| \geq 2$ for each $1 \leq i \leq r$. Let $|V|=n$ and let $I$ be the given {\sc{Multicoloured Independent Set}} instance.

We construct a \nmulticutname{} instance $I'=(G',\mathcal{T},n)$ as follows. We start with the graph $G$.
Then, for each $v \in V$ we create a vertex $v'$ and connect it to $v$. For each set $V_i$, we connect the vertices $\{v': v \in V_i\}$
into a path $P_i$ in an arbitrary order. We now add terminal pairs. Each terminal will be of degree one in the graph $G'$.

First, for each $v \in V$ we create a terminal $t_v$ connected to $v$ and we include in $\mathcal{T}$ all pairs $(t_v,t_u)$ for $u,v \in V$, $u \neq v$.
Second, for each $v \in V$ we create terminals $s_v$ and $s_v'$, connected to $v$ and $v'$ respectively, and include $(s_v,s_v')$ in $\mathcal{T}$.
Finally, for each set $V_i$, we create terminals $a_i$ and $b_i$, connected to two endpoints of the path $P_i$, and include $(a_i, b_i)$ in $\mathcal{T}$.
This finishes the construction of the instance $I'$.

First note that for each $(s,t) \in \mathcal{T}$, we have $N(s) \cap N(t) = \emptyset$, due to the assumption $|V_i| \geq 2$ for each $1 \leq i \leq r$.
Thus an assignment that takes $d_v = d_{v'} = 1/2$ for each $v \in V$ is a feasible solution to the LP-relaxation of cost $n$. Moreover, it is an optimal
solution, as $d_v + d_{v'} \geq 1$ for each $v \in V$ due to the terminal pair $(s_v, s_v')$. Thus $LP(I') = n$.

Assume $I$ is a YES-instance and let $X \subseteq V$ be an independent set of size $r$ in $G$. Take $X' = \{v': v \in X\}$ and $Y = (V \setminus X) \cup X'$. Clearly
$|Y| = n$. To see that $Y$ is a solution to the instance $I'$ observe that $V\setminus X$ is a vertex cover of $G$.

In the other direction, let $Y$ be a solution to the instance $I'$. $Y$ needs to include $v$ or $v'$ for each $v \in V$, due to the terminal pair $(s_v,s_v')$.
Thus $|Y|=n$ and $Y$ includes exactly one vertex from the set $\{v,v'\}$ for each $v \in V$. Moreover, for each $V_i$, if $u',v' \in Y$, $u,v \in V_i$, then
$Y$ does not disconnect $t_u$ from $t_v$. On the other hand, if $V_i \subseteq Y$, then $Y$ does not intersect $P_i$ and the pair $(a_i,b_i)$ is not disconnected by $Y$.
We infer that for each $1 \leq i \leq r$ there exists a vertex $v_i \in V_i$, such that $(V_i \setminus \{v_i\}) \cup \{v_i'\} \subseteq Y$.
Moreover, if $v_iv_j \in E$ for some $1 \leq i < j \leq r$, then the pair $(t_{v_i},t_{v_j})$ is not disconnected by $Y$. We infer that $\{v_i: 1 \leq i \leq r\}$ is an independent set in $G$,
and the instances $I$ and $I'$ are equivalent.
\end{proof}

\section{Conclusions}

In this paper, building upon work of Xiao \cite{xiao:multiway2010} and Guillemot
\cite{guillemot:iwpec2008}, we show that \nmcname{} is fixed-parameter tractable when
parameterized above two lower bounds: largest isolating cut and the cost of the optimal
solution of the LP-relaxation.
We also believe that our results may be of some importance in resolving the question
of an existence of a polynomial kernel for \nmcname{}.

One of the tools used in the parameterized complexity is the notion of important
separators introduced by Marx in 2004~\cite{marx-separators}.
From that time important separators were used for proving several problems
to be in FPT, including {\sc Multiway Cut}~\cite{marx-separators},{\sc Directed Feedback Vertex Set}~\cite{directed-fvs},
{\sc Almost 2-SAT}~\cite{razgon:icalp2008} and {\sc Multicut}~\cite{marx:multicut,thomasse:multicut}.
In this paper we show that in the \nmcname{} problem half-integral solutions
of the natural LP-relaxation of the problem can be even more useful than
important separators. Is it possible to use linear programming in other graph separation problems,
for example to obtain a $O^*(c^k)$ algorithm for {\sc{Directed Feedback Vertex Set}}?

We have shown that \nmulticutname{} parameterized above LP-relaxation is not in $XP$.
Is the edge-deletion variant similarly difficult?

\section*{Acknowledgements}

We thank Saket Saurabh for pointing us to~\cite{guillemot:iwpec2008}.

\bibliographystyle{plain}
\bibliography{multiway-above-lower-arxiv-20110708}

\end{document}